\documentclass[fleqn,10pt]{wlscirep}
\usepackage[utf8]{inputenc}
\usepackage[T1]{fontenc}

\usepackage{graphicx}
\usepackage{bm}
\usepackage{physics}
\usepackage{tikz}
\usepackage{amsmath}
\usepackage{amsthm}
\usepackage{amsfonts}
\usepackage{mathrsfs}
\usepackage{mathtools}
\usepackage{xcolor}
\usepackage[normalem]{ulem}
\newtheorem{dfn}{Definition}
\newtheorem{thm}{Theorem}
\newtheorem{prop}{Proposition}

\title{{Uniqueness of imaginarity-assisted exact transformation from real orthogonal operations to arbitrary unitary operations}}

\author[1,*]{Yasuaki Nakayama}
\author[2,1,3]{Yuki Takeuchi}
\author[1,3]{Seiseki Akibue}
\affil[1]{NTT Communication Science Laboratories, NTT Inc., 3-1 Morinosato Wakamiya, Atsugi, Kanagawa 243-0198, Japan}
\affil[2]{Information Technology R\&D Center, Mitsubishi Electric Corporation, 5-1-1 Ofuna, Kamakura, Kanagawa 247-8501, Japan}
\affil[3]{NTT Research Center for Theoretical Quantum Information, NTT Inc., 3-1 Morinosato Wakamiya, Atsugi, Kanagawa 243-0198, Japan}

\affil[*]{Corresponding author: Yasuaki Nakayama (email: yasuaki.nakayama@ntt.com)}

\begin{abstract}
The computational universality with an elementary gate set $\{H,CCZ\}$ can be transformed to the strict universality by using a maximally imaginary state $\ket{+i}$ and some non-imaginary ancillary qubits. From the viewpoint of operational resource theory, it would be intriguing to elucidate a resource for the universality transformation.
In this paper, we consider the exact universality transformation in which arbitrary real orthogonal matrices can be applied and a supplied resource state is used to simulate unitary operations exactly and deterministically. Within this operational model, we explore a necessary and sufficient condition for resource states to realize the universality transformation under free real operations. We show that $\ket{+i}$ is a unique resource state up to the free operations. Moreover, we obtain a stronger conclusion. If a given resource state cannot be used for the universality transformation, then realizable quantum gates are restricted to unitary matrices proportional to real orthogonal matrices. Therefore, we can tell that $\ket{+i}$ is unique (up to the free operations) not only as a state whose resource measure of imaginarity is maximal, but also as a state which empowers real operations with the ability to apply at least one non-real quantum gate (regardless of the magnitudes of its imaginary parts).
\end{abstract}

\begin{document}

\flushbottom
\maketitle
%
%
\thispagestyle{empty}


\section*{Introduction}


Quantum computer is thought to be more powerful than classical computer since quantum algorithms such as Shor's prime factoring algorithm \cite{Shor97}, Grover search \cite{Grover96}, and quantum simulation \cite{Feynman82, Lloyd96, Kassal08} are elucidated to be faster than known classical algorithms. In quantum computation \cite{Nielsen10, Watrous18}, first we input multiple qubits and then apply an appropriate unitary matrix to them in order to generate an output quantum state in which an answer of the problem is embedded. Finally, we read out the answer by measuring the whole or a part of the output qubits. For acting the unitary matrix, we do not implement it directly but decompose it to a combination of elementary quantum gates each of which acts on constant numbers of qubits.

There are two classes of universal quantum computation: strict universality and computational universality \cite{Aharonov03}. An elementary gate set is called strictly universal if it can approximate any unitary matrix with arbitrary precision. By using a strictly universal gate set, we can generate any desired quantum state, and hence it is thought that the strict universality is useful for computational tasks which need to prepare quantum states. On the other hand, an elementary gate set is called computationally universal if it can highly efficiently simulate any output probability distribution obtained by a quantum circuit (consisting of an arbitrary elementary gate set) with arbitrary precision. The computationally universal but not strictly universal gate set, which we simply call computationally universal in this paper, is sufficient to implement quantum computation whose output is classical such as prime factorization \cite{Shor97} but not sufficient for information processing which needs the creation of quantum states and computational tasks which are implemented by distributed quantum computers \cite{Renou21, Wu24}. An example of a strictly universal quantum gate set is $\{ H,\Lambda(S) \}$ \cite{Kitaev97}, while an example of a computationally universal one is $\{H,CCZ\}$ \cite{Shi02, Aharonov03} where the Hadamard gate $H$, the controlled-$S$ gate $\Lambda(S)$ and the controlled-controlled-$Z$ gate $CCZ$ are defined in the computational basis as follows:
\begin{equation}
    \begin{aligned}
        &H = \frac{1}{\sqrt{2}}\begin{bmatrix}
        1 & 1 \\
        1 & -1 \\
    \end{bmatrix}, \, S= \begin{bmatrix}
        1 & 0 \\
        0 & i \\
    \end{bmatrix}, \, Z= \begin{bmatrix}
        1 & 0 \\
        0 & -1 \\
    \end{bmatrix},\\
    &\Lambda(S) = \ket{0}\bra{0} \otimes I + \ket{1}\bra{1} \otimes S \, ,\\
    &CCZ = (I^{\otimes 2} - \ket{11}\bra{11}) \otimes I + \ket{11}\bra{11} \otimes Z \, ,
    \end{aligned}
\end{equation}
where $I$ is the two-dimensional identity gate.

{The computational universality of $\{H,CCZ\}$ follows from the previous works \cite{Aharonov03, Shi02}, while the absence of imaginary numbers in $\{H,CCZ\}$ explains why it is not strictly universal \cite{Takeuchi24}.} $\{H,CCZ\}$ can generate any real orthogonal matrix densely \cite{Amy23} but cannot generate all unitary matrices. A natural question is by which resource state we can realize the universality transformation: transformation from the computational universality to the strictly one.
This is important from a practical standpoint, as resourceful dynamics—such as unitary evolutions containing imaginarity—can be implemented by using the universality transformation as long as the resource state is stored as a static resource.
In the previous work \cite{Takeuchi24}, $\{H,CCZ\}$ is transformed to the strictly universal $\{H,\Lambda(S)\}$ as follows. $\Lambda(S)$ is constructed by using $\{ H,CCZ,S \}$ and one non-imaginary ancilla qubit $\ket{0}$, and $S$ is constructed by $\{H,CCZ\}$, some non-imaginary ancilla qubits and the eigenstate of the Pauli-$Y$ matrix $\ket{+i} = \frac{1}{\sqrt{2}} (\ket{0} + i \ket{1})$. Note that $\ket{+i}$ can be recovered after the generation of $S$, and thus one $\ket{+i}$ is sufficient even when we want to apply many $S$ gates.

However, it has not been clarified whether there are other resource states for the universality transformation.
Given the fact that $\{H,CCZ\}$ (with a non-imaginary ancillary state) is sufficient to approximate any real orthogonal matrix, we assume, for simplicity, that a given quantum gate set before the transformation is the set of all real orthogonal operators.
{We regard this as an idealized exact operational model motivated by the dense generation of real orthogonal matrices by $\{H,CCZ\}$, rather than as an exact equivalence to the conventional approximation-based setting using $\{H,CCZ\}$.}
We define a resource state which can be used for the simulation of any unitary matrix as universal resource. The previous work \cite{Takeuchi24} showed that $\ket{+i}$ is universal resource. The problem we try to solve is whether there are universal resources other than $\ket{+i}$. We call resources which can be used only for the simulation of real orthogonal matrices as zero resource. Apparently, there is expected to be some middle region between the universal resource set and the zero resource set. However, we prove that there is no such middle region. In other words, if $\rho$ is not zero resource, $\rho$ is universal resource. We also show that the universal resource is restricted to the maximally imaginary states, which can be transformed into any state with a real completely positive and trace preserving (CPTP) map. This also implies that states other than maximally imaginary states are zero resource. 
Since in imaginarity resource theory, it was known that $\ket{+i}$ is a unique maximally imaginary state up to the free operations \cite{Hicky18, Wu21PRL, Wu21PRA}, our results imply that $\ket{+i}$ is also unique for the universality transformation. Resource theory of imaginarity conventionally focuses on resource measure and state conversion, but this paper studies resource for unitary evolutions and paves the way for dynamical resource theory \cite{gour2020dynamicalresources} of imaginarity.

\section*{Results}





 

We derive the necessary and sufficient condition for an imaginary resource state $\rho$ to be usable as a resource for the universality transformation.
To consider a general setup, we do not assume that $\rho$ is pure.
Furthermore, we do not assume the catalytic transformation, i.e., {the supplied resource state need not be recovered after the universality transformation.}
The goal is to find the condition on $\rho$ and identify all resource states which can be used for the universality transformation.
{In this exact operational model, a fixed resource state $\rho$ is provided as an input for the simulation.}
The situation that $\rho$ can be used as a resource for simulating a unitary matrix $V$, provided that we can use only real orthogonal operators, is defined as follows.
\begin{dfn}
\label{oursituation}
A unitary matrix $V$ can be simulated using a resource $\rho$ if and only if it can be realized by appending an ancillary system in a real state, applying a unitary evolution represented by a real orthogonal matrix, and tracing out the environment system. More precisely, this condition can be written as
\begin{equation}\label{eq:exact_simulation_CPTP}
    \exists\, \text{a real orthogonal matrix } U \text{ and } \ket{\mathbf{0}}
    \text{ such that }
    \forall \ket{\psi},\quad
    \tr_E\left[
        U\left(
            \rho \otimes \ket{\mathbf{0}} \bra{\mathbf{0}}
        \otimes\ketbra{\psi}{\psi}
        \right)
        U^\dagger
    \right]
    =
    V\ketbra{\psi}{\psi}V^\dagger.
\end{equation}
Since the output state is pure, the environment system must be decoupled from the output system. Using this observation, we can equivalently define simulability as follows.
{A unitary matrix $V$ can be simulated by using a resource $\rho$ if and only if the following is satisfied:
\begin{equation}\label{eq:exact_simulation}
    \begin{aligned}
        &\exists \, \text{a real orthogonal matrix } U \text{ and } \ket{\mathbf{0}},
        \, \forall \ket{\psi}, \, \exists \sigma_E(\ket{\psi}),\\
        &U (\rho \otimes \ket{\mathbf{0}} \bra{\mathbf{0}}
        \otimes \ket{\psi} \bra{\psi}) U^\dagger
        =
        \sigma_E(\ket{\psi})
        \otimes V \ket{\psi} \bra{\psi} V^\dagger \, .
    \end{aligned}
\end{equation}
Here, $\ket{\mathbf{0}}$ denotes an input non-imaginary ancillary state,
and $\sigma_E(\ket{\psi})$ is the density operator of all non-target output systems.
We show in the Methods that $\sigma_E(\ket{\psi})$ is independent of $\ket{\psi}$.}
\end{dfn} 
Note that a CPTP map admits a Kraus representation with real Kraus operators if and only if it can be realized by appending an ancillary system in a real state, applying a unitary evolution represented by a real orthogonal matrix, and tracing out the environment. For completeness, we provide a proof of this fact in the Methods. The proof also shows that the dimension of the ancillary system can be chosen to be at most $d^2$, where $d$ is the dimension of the Hilbert space on which $V$ acts.

\noindent Based on Definition~\ref{oursituation}, we define zero and universal resources as follows:
\begin{dfn}
\label{zeroanduniversalresource}
    If a unitary matrix {$e^{i \theta}V$} can be simulated with $\rho$ for some real number $\theta$, we call $\rho$ as $V$-resource. We say that $\rho$ is zero resource if{, for every unitary matrix $V$ for which $\rho$ is $V$-resource, $V$ is proportional to a real orthogonal matrix, i.e., $V = e^{i\phi} O$ for some real number $\phi$ and some real orthogonal matrix $O$}.
    On the other hand, if $\rho$ is $V$-resource for any unitary matrix $V$, then we call it universal resource.
\end{dfn}
{In this paper, we call a unitary matrix $V$ a non-real gate if $V$ is not proportional to any real orthogonal matrix.}
By using Definition~\ref{zeroanduniversalresource}, our main result is summarized as follows:
{The following theorem is established within the exact operational model defined above.}
\begin{thm}[Main theorem of this paper]
    $\rho$ is not zero resource if and only if $\rho$ is universal resource.
\end{thm}

\begin{proof}
The sufficiency is straightforward. We prove the necessity below.
First, we explain the outline of the proof. We show the fact (1) that if $\rho$ is not zero resource, $\tr[\rho \rho^\ast] =0$ is satisfied. This is the main part of the proof and we explain this later. Then we can easily prove the fact (2) that $\tr[\rho \rho^\ast] =0$ is equivalent to $||\rho -\rho^\ast||_1 =2$, which is explained in Methods.
We use the fact (3) discovered in the previous work \cite{Wu21PRA} that the fidelity {for transforming $\rho$ by real operations into the maximally imaginary state $\ket{\hat{+}} = \frac{1}{\sqrt{2}} (\ket{0} + i \ket{1})$} can be written as $F_{I} (\rho) = \frac{1}{2} + \frac{1}{4} ||\rho - \rho^\ast ||_{1}$.
It can be proved that if $||\rho -\rho^\ast||_1 =2$ is satisfied, $\rho$ can be exactly transformed to $\ket{\hat{+}}$ by a real operation.
We can prove the fact (4) that $\rho$ is universal resource if and only if $\rho$ can be transformed to {$\ket{\hat{+}}$} by a real CPTP map. We prove this in Methods.
We also show in the Methods that this real CPTP map can be realized by adding {a real ancillary system, applying a real orthogonal matrix, and tracing out the environment system.}
By using this fact, we can tell that $\rho$ is universal resource. As a result, we obtain Theorem 1.

{We show that if $\rho$ is not zero resource, $\rho$ and the residual state $\sigma_E$ must satisfy the condition
$\tr[\rho \rho^\ast] = \tr[\sigma_E \sigma_E^\ast] = 0$.
First, we can say that $\sigma_E(\ket{\psi})$ is independent of $\ket{\psi}$. This is proved in the Methods, and we denote it simply by $\sigma_E$ hereafter. Since \eqref{eq:exact_simulation} holds for any $\ket{\psi}$, we first apply it to $\ket{\psi^\ast}$.
By taking the complex conjugate of the resulting equation, and using the fact that $U$ and $\ket{\mathbf{0}}\bra{\mathbf{0}}$ are real and that $\sigma_E$ is independent of $\ket{\psi}$, we obtain
\begin{equation} \label{eq:complex_conjugate}
    U (
    \rho^{\ast}
    \otimes \ket{\mathbf{0}} \bra{\mathbf{0}}
    \otimes \ket{\psi} \bra{\psi}
    ) U^\dagger
    = \sigma_E^\ast
    \otimes
    V^{\ast} \ket{\psi} \bra{\psi} V^T \, .
\end{equation}
}
By taking the Hilbert-Schmidt inner product of \eqref{eq:exact_simulation} and \eqref{eq:complex_conjugate}, we can derive the following condition.{\begin{equation}\label{eq: condition of density matrix1}
     \exists \sigma_E, \, \forall \ket{\psi}, \, \,
     \mathrm{tr} [\rho \rho^\ast]
     = \mathrm{tr} [\sigma_E \sigma_E^\ast]
     \cdot | \bra{\psi} V^T V \ket{\psi} |^2 \, .
\end{equation}
Here, since we have shown that $\sigma_E$ is independent of $\ket{\psi}$, we interchanged the order of the universal quantifier $\forall \ket{\psi}$ and the existential quantifier $\exists \sigma_E$.}
Since $| \bra{\psi} V^T V \ket{\psi} |^2 = 1$ is satisfied when $\ket{\psi}$ is an eigenstate of the unitary matrix $V^T V$, we can derive the equivalent condition to \eqref{eq: condition of density matrix1}:{\begin{equation}\label{eq:simulatability1}
    \Big( \forall \ket{\psi}, \,
    \mathrm{tr} [\rho \rho^\ast]
    = \mathrm{tr} [\rho \rho^\ast]
    \cdot | \bra{\psi} V^T V \ket{\psi} |^2\Big)
    \wedge
    \Big(
    \mathrm{tr} [\rho \rho^\ast]
    = \mathrm{tr} [\sigma_E \sigma_E^\ast]
    \Big) \, .
\end{equation}
}
{If $\tr[\rho\rho^\ast]=0$, the second condition in \eqref{eq:simulatability1} implies that $\tr[\sigma_E\sigma_E^\ast]=0$.} Indeed, if $\tr [\rho \rho^\ast] \neq 0$, $V$ is restricted to {unitary matrices proportional to} real orthogonal matrices, which means that $\rho$ is zero resource. We prove this as follows.
When $\tr [\rho \rho^\ast] \neq 0$, we derive $| \bra{\psi} V^T V \ket{\psi} |^2 = 1$ from \eqref{eq:simulatability1}. This can be deformed to the following equation:
\begin{equation} \label{eq:global phase1}
    \forall \ket{\psi}, \, \bra{\psi} V^T V \ket{\psi} = e^{i \eta(\ket{\psi})}
\end{equation}
where $\eta = \eta(\ket{\psi})$ depends on $\ket{\psi}$ in general. We show that $\eta$ is indeed independent of $\ket{\psi}$. By expanding $\ket{\psi}$ by {an orthonormal eigenbasis $\{\ket{i}\}_{i=0}^{d-1}$} of $V^T V$, we get $\ket{\psi} = \sum_{i} a_{i} \ket{i}$ with $\sum_{i} |a_{i}|^2 =1$ {where these basis vectors are not necessarily computational basis vectors}.
Since $V^T V$ is a unitary matrix, we can transform \eqref{eq:global phase1} to the following equation if we define $b_{j}$ by $V^T V \ket{j} = e^{i{b_{j}}} \ket{j}$.
\begin{equation}
    \forall \ket{\psi}, \, \sum_ j |a_j|^2 e^{i b_{j}} = e^{i \eta (\ket{\psi})}.
\end{equation}
\begin{enumerate}
    \item When $\ket{\psi} = \ket{{i= 0}}$, we obtain
    \begin{equation}
        e^{i b_{{0}}} = e^{i \eta (\ket{{0}})} \, .
    \end{equation}
    \item When $\ket{\psi}$ is {another eigenstate,} we write $\ket{\psi} = \ket{{i=x}}$ {where $x \neq 0$, i.e. $\ket{\psi} \neq e^{iy} \ket{0}$} for any $y \in \mathbb{R}$,
    \begin{equation}
        e^{i b_{x}} = e^{i \eta (\ket{x})} \, .
    \end{equation}
    \item When $\ket{\psi} = \frac{1}{\sqrt{2}} (\ket{{0}} + \ket{x})$,
    \begin{equation}
        \begin{aligned}
            \frac{1}{2} (e^{i b_{{0}}} + e^{i b_{x}}) &= \frac{1}{2} (e^{i \eta (\ket{{0}})} + e^{i \eta(\ket{x})} )\\
            &= e^{i \eta (\ket{\psi})}.
        \end{aligned}
     \end{equation}
\end{enumerate}
It is only when two points are the same on a unit circle that the middle point between the two points is on a unit circle. Therefore we can conclude that $\eta$ is independent of $\ket{\psi}$:
\begin{equation}
    e^{i \eta (\ket{{0}})} = e^{i \eta (\ket{x})} = e^{i \eta (\ket{\psi})}.
\end{equation}
For all $\ket{\psi}$, there exists some real constant $\eta$ such that $\bra{\psi} V^T V \ket{\psi} = e^{i \eta}$. Therefore we get
\begin{equation}
    V^T V = e^{i \eta} I.
\end{equation}
If we set $V^\prime = e^{- i \eta /2} V$, $V^{\prime T} V^\prime = I$ is satisfied{. Since $V^\prime$ is a unitary matrix, $V^{\prime T} V^\prime = I$ implies $V^{\prime T} = V^{\prime -1} = V^{\prime \dagger}$} and hence {$V^\prime = V^{\prime *}$, so} $V^\prime$ is a real orthogonal matrix. {Thus, if $\tr[\rho \rho^\ast] \neq 0$, every unitary matrix $V$ simulable with $\rho$ is proportional to a real orthogonal matrix. Therefore, $\rho$ is zero resource. By taking the contrapositive, we can tell that if $\rho$ is not zero resource,
$\tr[\rho \rho^\ast] = \tr[\sigma_E \sigma_E^\ast]=0$
must be satisfied.}

Here, we explain the detail of the fact (3). In the paper of imaginarity resource \cite{Wu21PRL, Wu21PRA}, real states $\mathcal{R} = \{ \rho : \bra{m} \rho \ket{n} \in \mathbb{R} \}$ are defined as free states, where the basis is fixed. In addition, real operations $\Lambda[\rho] = \sum_{j} K_{j} \rho K_{j}^\dagger$ where $\bra{m} K_{j} \ket{n} \in \mathbb{R}$ are defined.
{The maximally imaginary state \cite{Hicky18, Wu21PRA} is known to be $\ket{\hat{+}} = \frac{1}{\sqrt{2}} (\ket{0} + i \ket{1})$.} The fidelity to convert $\rho$ to $\sigma$ by a real operation is defined as $F(\rho \rightarrow \sigma) \coloneq \max_\Lambda \{F(\Lambda[\rho], \sigma)\}$, and the imaginarity fidelity is derived as $F_{I} (\rho) \coloneq F(\rho \rightarrow \ket{\hat{+}} \bra{\hat{+}}) = \frac{1}{2} + \frac{1}{4} ||\rho - \rho^\ast ||_{1}$.
We have shown by (1) and (2) that if $\rho$ is not zero resource, $||\rho - \rho^\ast||_1 = 2$ is satisfied, and thus the imaginarity fidelity $F_{I}$ becomes 1. This means that $\rho$ can be exactly transformed to $\ket{\hat{+}}$ by a real operation. It is guaranteed by using the discussion of the previous works \cite{Wu21PRA, Hicky18} that this real operation can be implemented by a real CPTP map{. As shown in Methods,  this real CPTP map can be realized by adding a real ancillary system, applying a real orthogonal matrix, and tracing out an environment system. See the Methods for the explanation.}
\end{proof}

Note that if $\rho$ is not universal resource, $\rho$ is zero resource even if it possesses some imaginarity. It might be naively expected that there exists intermediate resource between universal and zero resources, but it is surprisingly shown that there is no such intermediate resource as shown in figure \ref{fig:intermediate resource}. We remark that we have assumed the exact simulation which imposes equality in \eqref{eq:exact_simulation}. If we consider approximate simulation, then the definition of universal resource will be changed, and intermediate resource could occur but this problem is beyond the scope of this paper.
{Recently, Hokkyo and Tajima derived a quantitative bound for approximate implementations of non-real unitary operations, which quantifies the no-go result of the present work \cite{Hokkyo26}.}

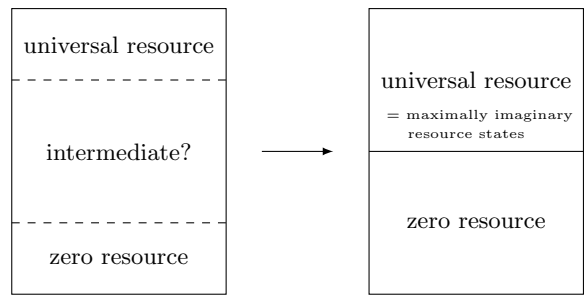
\begin{figure}
\hspace*{-5mm}
\centering
\begin{tikzpicture}[x=0.95cm,y=0.95cm]
    \draw (0,0) rectangle (3,4);
    \draw[dashed](0,1)--(3,1);
    \draw[dashed](0,3)--(3,3);
    \draw[-latex] (3.5,2)--(4.5,2);
    \draw (5,0) rectangle (8,4);
    \draw(5,2)--(8,2);
    \draw (1.5,0.5) node{zero resource};
    \draw (1.5,2) node{intermediate?};
    \draw (1.5,3.5) node{universal resource};
    \draw (6.5,1) node{zero resource};
    \draw (6.5,3) node{universal resource};
    \draw (6.55,2.5) node[font=\scriptsize]{$=$ maximally imaginary};
    \draw (6.36,2.25) node[font=\scriptsize]{resource states};
\end{tikzpicture}
\caption{Naively, there is expected to be an intermediate region between universal resource and zero resource, but {in the exact operational model considered in this paper,} there is actually no such region and there are only two types: universal and zero resources. $\mathbf{Left}:$ Naive expectation. $\mathbf{Right}:$ Actual situation.}
\label{fig:intermediate resource}
\end{figure}

In the following proposition, we show the special case of using a single-qubit imaginary resource state.
\begin{prop}
    If $\rho$ is a single-qubit imaginary resource state that satisfies $\tr[\rho \rho^\ast] = 0$, then $\rho$ is restricted to $\ket{+i} \bra{+i}$ or $\ket{-i} \bra{-i}$, where $\ket{\pm i}$ are single-qubit imaginary states $\frac{1}{\sqrt{2}} (\ket{0} \pm i \ket{1})$, respectively.
\end{prop}
\begin{proof}
    The density matrix of a single-qubit mixed state is expressed as follows: $\rho = \frac{1}{2} \left( I + \vec{r} \cdot \vec{\sigma} \right)$ where we denote $I$ as the $2 \times 2$ identity matrix and $\vec{\sigma} = [X, Y, Z]^T$ as the vector whose components are the Pauli matrices. We define the Bloch vector $\vec{r}$ as $\vec{r} = [x,y,z]^T = \left[\mathrm{tr} [\rho X],\mathrm{tr} [\rho Y],\mathrm{tr} [\rho Z] \right]^T$. We remark that $\rho$ is a single-qubit density matrix if and only if $\|\vec{r}\| \leq 1$. Pure states exist on the Bloch sphere $\|\vec{r}\| =1$, and mixed states exist completely inside it, i.e., $\|\vec{r}\| <1$. We can calculate as
    \begin{equation}
        \tr [\rho \rho^\ast] = \frac{1}{2} (1+x^2 -y^2 +z^2) =0 \iff y^2 -1 = x^2 +z^2.
    \end{equation}
    Since $x^2 + z^2$ is non-negative, we can conclude that $y^2 -1 \geq 0$. Since $-1 \leq y \leq 1$ is satisfied, the allowed value of $y$ is only $y= \pm 1$. The correspondence of each value is $y=1 \iff \rho = \ket{+i} \bra{+i}$ and $y=-1 \iff \rho = \ket{-i} \bra{-i}$.
\end{proof}
\noindent We remark that $\ket{-i}$ can be created by acting the Hadamard gate to $\ket{+i}$.

\section*{Discussion}


We have identified every imaginarity resource state for {the exact universality transformation considered in this paper}, and conclude that if we do not prepare a maximally imaginary state, we can only simulate {unitary matrices proportional to} real orthogonal matrices.
Especially, if we consider using a single qubit, we can only use $\ket{\pm i}$ for the universality transformation and can obtain $\ket{+i}$ or $\ket{-i}$ as the resultant state of the transformation.
This resolves an open problem in the previous work \cite{Takeuchi24}, whether other states rather than $\ket{+i}$ can be used as a resource for {the exact universality transformation considered in this paper}.

In resource theory of imaginarity, it is known that $\ket{+i}$ is a unique maximally imaginary state up to free real operations \cite{Hicky18, Wu21PRL, Wu21PRA}. We have shown that $\ket{+i}$ is also unique for realizing at least one non-real quantum gate.

If we consider the $S$ gate, the necessary input resource state for generating the $S$ gate (i.e., the cost for generating a unitary operator) is $\ket{+i}$, and the most resourceful state which can be created by the $S$ gate (i.e., the power of a unitary operator) is also $\ket{+i}$. However, if we consider a unitary operator with small imaginary components, the cost for generating it and its power may differ in general. In entanglement theory, the difference between the cost and power has been studied \cite{Chen16, Soeda13}. Studying this difference in resource theory of imaginarity would be an interesting future work.

\section*{Methods}


{\subsection*{Proof of the $\ket{\psi}$-independence of the residual state $\sigma_E$}
\begin{proof}
An arbitrary $r$-qubit state $\ket{\psi}$ can be expressed as
\begin{equation}
    \ket{\psi}
    = c_{00 \cdots 0} \ket{00 \cdots 0}
    + \sum_{x \neq 00 \cdots 0} c_x \ket{x}
    = c_0 \ket{\phi_0} + c_1 \ket{\phi_1} \, .
\end{equation}
Here, we set $\ket{\phi_0}=\ket{00\cdots0}$ and, when $c_1\neq0$,
$\ket{\phi_1}=\frac{1}{c_1}\ket{\phi^\prime}$, where
$\ket{\phi^\prime}=\sum_{x\neq00\cdots0}c_x\ket{x}$ and
$c_1=\sqrt{\braket{\phi^\prime}{\phi^\prime}}$.

There are three cases to consider:
(1) $c_1=0$, (2) $c_0=0$, and
(3) $c_0\neq0 \wedge c_1\neq0$.

In case (1), it is obvious that
$\sigma_E(\ket{\phi_0})=\sigma_E(\ket{\psi})$ since
$\ket{\psi}\bra{\psi}=\ket{\phi_0}\bra{\phi_0}$.

In case (3), $\ket{\psi}$ is neither parallel to
$\ket{\phi_0}$ nor to $\ket{\phi_1}$.
If we set
$\ket{\phi_2}=c_0\ket{\phi_0}-c_1\ket{\phi_1}$,
then we obtain
\begin{equation}
    \ket{\psi}\bra{\psi}
    = \alpha \ket{\phi_0}\bra{\phi_0}
    +\beta \ket{\phi_1}\bra{\phi_1}
    -\ket{\phi_2}\bra{\phi_2},
\end{equation}
where $\alpha=2|c_0|^2\neq0$ and
$\beta=2|c_1|^2\neq0$.
Using the linearity of the transformation and
Definition~\ref{oursituation}, we obtain
\begin{equation}\label{eq.1}
\begin{aligned}
&\alpha \sigma_E(\ket{\psi})
 \otimes V\ket{\phi_0}\bra{\phi_0}V^\dagger
+\beta \sigma_E(\ket{\psi})
 \otimes V\ket{\phi_1}\bra{\phi_1}V^\dagger
-\sigma_E(\ket{\psi})
 \otimes V\ket{\phi_2}\bra{\phi_2}V^\dagger\\
&= \sigma_{E}(\ket{\psi}) \otimes V \ket{\psi} \bra{\psi} V^\dagger = 
U\left(
\rho\otimes\ket{\mathbf 0}\bra{\mathbf 0}
\otimes\ket{\psi}\bra{\psi}
\right)U^\dagger\\
&=
\alpha \sigma_E(\ket{\phi_0})
 \otimes V\ket{\phi_0}\bra{\phi_0}V^\dagger
+\beta \sigma_E(\ket{\phi_1})
 \otimes V\ket{\phi_1}\bra{\phi_1}V^\dagger
-\sigma_E(\ket{\phi_2})
 \otimes V\ket{\phi_2}\bra{\phi_2}V^\dagger .
\end{aligned}
\end{equation}

Since $c_0$ and $c_1$ are nonzero,
$\ket{\phi_0}\bra{\phi_0}$,
$\ket{\phi_1}\bra{\phi_1}$, and
$\ket{\phi_2}\bra{\phi_2}$ are linearly independent.
Since $V$ is unitary,
$V\ket{\phi_0}\bra{\phi_0}V^\dagger$,
$V\ket{\phi_1}\bra{\phi_1}V^\dagger$, and
$V\ket{\phi_2}\bra{\phi_2}V^\dagger$
are also linearly independent.
Therefore, by comparing their operator-valued coefficients
in Eq.~\eqref{eq.1}, we obtain
\begin{equation}
    \sigma_E(\ket{\psi})
    =\sigma_E(\ket{\phi_0})
    =\sigma_E(\ket{\phi_1})
    =\sigma_E(\ket{\phi_2}).
\end{equation}
Therefore, the statement is proved in case (3).

In case (2), for any $\ket{\phi_1}$ satisfying
$\braket{\phi_0}{\phi_1}=0$, we set
$\ket{\psi}=\frac{1}{\sqrt{2}}
(\ket{\phi_0}+\ket{\phi_1})$.
Since $\ket{\psi}$ satisfies the condition of case (3),
we have
$\sigma_E(\ket{\phi_0})
=\sigma_E(\ket{\psi})
=\sigma_E(\ket{\phi_1})$.
Therefore, the statement is also proved in case (2).

In all cases,
$\sigma_E(\ket{\phi_0})=\sigma_E(\ket{\psi})$.
Therefore, we conclude that the residual state
$\sigma_E$ is independent of $\ket{\psi}$.
\end{proof}
}

\subsection*{Proof of the equivalence of $\tr[\rho \rho^\ast] = 0$ and $||\rho - \rho^\ast||_1 = 2$}
\begin{proof}
We show that $\tr[\rho \sigma] = 0$ is equivalent to $||\rho - \sigma||_1 =2$.

$\longrightarrow$ ) By expanding $\rho$ and $\sigma$ by their orthonormal bases respectively, we can write down $\rho = \sum_{i} p_{i} \ket{e_{i}} \bra{e_{i}}$ and $\sigma = \sum_{j} q_{j} \ket{f_{j}} \bra{f_{j}}$ where $p_{i}, q_{j} \geq 0$. Then we can write $\tr[\rho \sigma] = \sum_{ij} p_{i} q_{j} |\braket{e_{i}}{f_{j}}|^2 = 0$. Therefore, we can derive that $\forall i,j, \, p_{i} q_{j} |\braket{e_{i}}{f_{j}}|^2 =0$. If we focus on the supports: $\text{supp} [\rho] = \text{Span} \{ \ket{e_{i}}: p_{i} > 0 \}$ and $\text{supp} [\sigma] = \text{Span} \{ \ket{f_{j}}: q_{j} > 0 \}$, it follows that if $p_{i} >0$ and $q_{j} >0$, then $\braket{e_{i}}{f_{j}} =0$. It is concluded that $\rho - \sigma$ can be transformed to block-diagonal by some unitary matrix: $U (\rho - \sigma) U^\dagger = diag[d_{1} \cdots d_k \, d^\prime_{1} \cdots d^\prime_{k^\prime}]$ where $\{d_{i}\}$ and $\{d^\prime_{j}\}$ are the eigenvalues of $\rho$ and $-\sigma$ respectively. Thus, we can conclude $||\rho -\sigma||_{1} = \sum_{i}|d_{i}| + \sum_{j} |d^\prime_{j}| = 2$.

$\longleftarrow$ ) {By using the Holevo-Helstrom theorem (Theorem 3.4 in Ref. \cite{Watrous18}), we can tell that for any density operators $\rho_0$ and $\rho_1$,
\begin{equation}
    \max_{0 \leq M \leq I} \left[ \langle M, \rho_0 \rangle + \langle I-M, \rho_1 \rangle \right] = 1 + \frac{1}{2}||\rho_0 - \rho_1||_1
\end{equation}
where $M$ is a POVM operator satisfying $0 \leq M \leq I$. Therefore we can prove
\begin{equation}
    2= ||\rho -\sigma||_1 = \max_{0 \leq M \leq I} 2 \tr [(\rho- \sigma) M] \, .
\end{equation}
}
Since $0 \leq \tr[\pi M] \leq 1$ is satisfied for any density matrix $\pi$, we can conclude that $\tr[\rho M] =1$ and $\tr[\sigma M] =0$. Since $\rho$ lies in the Hilbert space orthogonal to $\mathbf{1} - M$ and $\sigma$ lies in the Hilbert space {orthogonal} to $M$, {we can conclude that} $\rho$ and $\sigma$ are orthogonal to each other: $\tr[\rho \sigma] =0$.
\end{proof}

\subsection*{Proof of the equivalence between universal resource and maximally imaginary states, which can be transformed to $\ket{\hat{+}}$ by a real CPTP map.}
\begin{proof}
$\longrightarrow$) If we have a universal resource, by using it, we can simulate any unitary operator including the $S$ gate. Therefore, we can generate $\ket{\hat{+}}$ by acting the $S$ gate to $\ket{+} \equiv (\ket{0}+\ket{1})/\sqrt{2}=H\ket{0}$.
{More explicitly, by taking the input state in Definition~\ref{oursituation} as $\ket{+}$, exact simulation of the $S$ gate gives
\begin{equation}
U\left(
\rho\otimes\ket{\mathbf{0}}\bra{\mathbf{0}}
\otimes\ket{+}\bra{+}
\right)U^\dagger
= \sigma_E\otimes\ket{\hat{+}}\bra{\hat{+}} .
\end{equation}
By tracing out the environment system $E$, this procedure defines a CPTP map from $\rho$ to $\ket{\hat{+}}\bra{\hat{+}}$. Since $U$, $\ket{\mathbf{0}}$, and $\ket{+}$ are real, the Kraus operators of this CPTP map can be chosen real. Therefore, $\rho$ can be transformed to $\ket{\hat{+}}$ by a real CPTP map.}

$\longleftarrow$) Suppose that $\rho$ can be transformed to $\ket{\hat{+}}\bra{\hat{+}}$ by a real CPTP map. As shown in the following subsection, this real CPTP map can be realized by adding {a real ancillary system}, applying a real orthogonal matrix, and tracing out the environment system.
The $S$ gate and $\Lambda(S)$ gate can be implemented by the circuits of FIG. 3 and FIG. 1 of \cite{Takeuchi24}, respectively with $\ket{0}$ and $\ket{\hat{+}}$ as inputs. We can perform free operations including the rotation operator around the $y$ axis $R_{y}(\theta)$. We can also implement the rotation operator around the $x$ axis as $R_{x} (\theta) = S^\dagger R_{y}(\theta) S$. From the rotation operators around the $x$ axis and $y$ axis, we can generate any unitary operator in $SU(2)$. We can also implement the controlled-$Z$ gate as $\Lambda(Z) = \Lambda(S)^2$. Therefore, we can generate any unitary operator in $SU(2^m)$ for any positive number $m$. As a result, $\rho$ is universal resource.
\end{proof}

{\subsection*{Remark on the construction of a real operation from $\rho$ to $\ket{\hat{+}} \bra{\hat{+}}$}
}
In the proof of Theorem 1 of our paper, we use the expression of imaginarity fidelity shown in the previous work \cite{Wu21PRA}:
\begin{equation}
    F_{I}(\rho) = F(\rho \rightarrow \ket{\hat{+}} \bra{\hat{+}}) = \frac{1}{2} (1 + \mathscr{I}_{R} (\rho)) = \frac{1}{2} + \frac{1}{4} ||\rho - \rho^T||_1 \, .
\end{equation}
Here, the third equality was proved in Proposition 5 and the second equality was proved in Theorem 3 in \cite{Wu21PRA}. In resource theory of imaginarity, real states are defined as $\mathcal{R} = \{ \rho: \bra{m} \rho \ket{n} \in \mathbb{R} \}$ in some fixed basis $\{ \ket{m} \}$, and a real operation $\Lambda$ is defined with real Kraus operators $K_{i}$ as $\Lambda[\rho] = \sum_{i} K_{i} \rho K_{i}^\dagger$ where $\bra{m} K_{i} \ket{n} \in \mathbb{R}$. The robustness of imaginarity is defined as $\mathscr{I}_{R}(\rho) = \min_{\tau} \{ s \geq 0 : \frac{\rho + s \tau}{1+s} \in \mathcal{R}\}$ where $\tau$ is taken over all quantum states. $F(\rho \rightarrow \ket{\hat{+}} \bra{\hat{+}})$ is defined as
\begin{equation}
    F(\rho \rightarrow \ket{\hat{+}} \bra{\hat{+}} ) = \max_{\Lambda} \{ F(\Lambda[\rho], \ket{\hat{+}} \bra{\hat{+}}) \}
\end{equation}
where the right-hand side is the fidelity maximized over real operations $\Lambda$. The construction of an actual real operation which realizes this maximization is shown in Theorem 3 of \cite{Wu21PRA}. 
For completeness, we first review the construction. This also clarifies the required dimension of the ancillary system.

\newpage
In the proof of Theorem 3 of \cite{Wu21PRA}, {it was shown that} for any real operation $\Lambda$, the fidelity is upper bounded as
\begin{equation}
    F(\Lambda[\rho], \ket{\hat{+}}\bra{\hat{+}}) = \bra{\hat{+}} \Lambda[\rho] \ket{\hat{+}} \leq \frac{1}{2} (1 + \mathscr{I}_{R}(\rho)) \,.
\end{equation}
We assume that $\rho$ is a state in $d$-dimensional Hilbert space.
{In the proof of Theorem 3 of \cite{Wu21PRA}, a real orthogonal matrix $O$ is first chosen such that $O\,{\rm Im}(\rho)\,O^T$ has the canonical block-diagonal form. Then the following concrete real Kraus operators are applied to $O\rho O^T$.} When the dimension $d$ is even,
\begin{equation}
    K_{m} = \ket{1} \bra{2m} + \ket{0} \bra{2m+1} \quad m=0,1, \cdots , d/2 -1 \,.
\end{equation}
When $d$ is odd, for $m \leq \lfloor d/2 \rfloor -1$ we define $K_{m}$ as above and we define
\begin{equation}
    K_{\lfloor d/2 \rfloor} = \ket{0} \bra{d-1} \,.
\end{equation}

{The Kraus operators acting directly on the original state $\rho$ are therefore $L_m=K_mO$. The corresponding real CPTP map is
\begin{equation}
    \widetilde{\Lambda}[X]
    =\sum_{m=0}^{r-1}L_mXL_m^\dagger
    =\sum_{m=0}^{r-1}K_mOXO^TK_m^\dagger \,,
\end{equation}
where $r=\lceil d/2\rceil$. Since $O$ and all $K_m$ are real and $\sum_mK_m^\dagger K_m=I$, we have $\sum_mL_m^\dagger L_m=I$, and hence $\widetilde{\Lambda}$ is a real CPTP map. When $F_I(\rho)=1$, the construction in \cite{Wu21PRA} gives $\widetilde{\Lambda}[\rho]=\ket{\hat{+}}\bra{\hat{+}}$.

We next show that any real operation 
\[\Lambda(\cdot)=\sum_{m=0}^{r-1}L_m\cdot L_m^\dagger\]
can be realized by adding a real ancillary system, applying a real orthogonal matrix and tracing out an environment system.
Let $A$ and $B$ denote the $d_A$-dimensional input and $d_B$-dimensional output systems of $\Lambda$, respectively.
Choose a positive integer $q$ and $R$ such that $d_Aq=d_BR$ and $R\geq r$, and pad the Kraus representation by setting $L_m=0$ for $r\leq m\leq R-1$.

Let $C$ be a $q$-dimensional ancillary system initialized in $\ket{\mathbf{0}}_C$, and let $E$ be an environment system described by $R$-dimensional Hilbert space with computational basis $\{\ket{m}_E\}_{m=0}^{R-1}$. Define a real Stinespring isometry $W:A\rightarrow B\otimes E$ by
\begin{equation}
    W\ket{\phi}_A=\sum_{m=0}^{R-1}L_m\ket{\phi}_A\otimes\ket{m}_E \,.
\end{equation}
Since all $L_m$ are real and $W^\dagger W=\sum_{m=0}^{R-1}L_m^\dagger L_m=I_A$, $W$ is a real isometry.

The total input and output dimensions are the same, $\dim(A\otimes C)=d_Aq=d_BR=\dim(B\otimes E)$. Therefore, the real orthonormal vectors defined by $W$ can be completed to a real orthonormal basis of the $d_BR$-dimensional output space. Thus, there exists a real orthogonal matrix $U:A\otimes C\rightarrow B\otimes E$ such that
\begin{equation}
    U\left(\ket{\phi}_A\otimes\ket{\mathbf{0}}_C\right)=W\ket{\phi}_A
\end{equation}
for any $\ket{\phi}_A$. Therefore,
\begin{equation}
    \Lambda[X]=\tr_E\left[U\left(X\otimes\ket{\mathbf{0}}\bra{\mathbf{0}}_C\right)U^\dagger\right] \,.
\end{equation}
This means that the real operation $\Lambda$ can be realized by the combination of adding a real ancillary system, the evolution by a real orthogonal matrix and tracing out the environment system.}

Thanks to the Choi--Jamio\l{}kowski isomorphism, the number $r$ of {real Kraus operators} can be chosen to satisfy $r \le d_A d_B$. In particular, when $d_A=d_B=d$, the dimension $q$ of the ancillary system can be chosen to be at most $d^2$.

When $F_I(\rho)=1$ and the dimension $d$ of the Hilbert space on which $\rho$ acts is even, the CPTP map $\widetilde{\Lambda}$ that transforms $\rho$ into $\ketbra{\hat{+}}{\hat{+}}$ can be realized without an ancillary system. Indeed, we may choose $q=1$, $d_A=d$, $d_B=2$, and $R=\frac{d}{2}$.

\section*{Data availability}
All data generated or analysed during this study are included in this published article.

\bibliography{catalyst}

@ARTICLE{Takeuchi24,
   author       = {{Y. Takeuchi}},
   title        = {Catalytic Transformation from Computationally Universal to Strictly Universal Measurement-Based Quantum Computation},
   year         = {2024},
   journal      = {Phys. Rev. Lett.},
   volume       = {133},
   pages        = {050601},
   doi          = {10.1103/PhysRevLett.133.050601},
}

@ARTICLE{Wu21PRL,
   author       = {{K.-D. Wu} and {T. V. Kondra} and {S. Rana} and {C. M. Scandolo} and {G.-Y. Xiang} and {C.-F. Li} and {G.-C. Guo} and {A. Streltsov}},
   title        = {Operational resource theory of imaginarity},
   year         = {2021},
   journal      = {Phys. Rev. Lett.},
   volume       = {126},
   pages        = {090401},
   doi          = {10.1103/PhysRevLett.126.090401},
}

@ARTICLE{Wu21PRA,
   author       = {{K.-D. Wu} and {T. V. Kondra} and {S. Rana} and {C. M. Scandolo} and {G.-Y. Xiang} and {C.-F. Li} and {G.-C. Guo} and {A. Streltsov}},
   title        = {Resource theory of imaginarity: Quantification and state conversion},
   year         = {2021},
   journal      = {Phys. Rev. A},
   volume       = {103},
   pages        = {032401},
   doi          = {10.1103/PhysRevA.103.032401},
}

@ARTICLE{Hicky18,
   author       = {{A. Hickey} and {G. Gour}},
   title        = {Quantifying the imaginarity of quantum mechanics},
   year         = {2018},
   journal      = {J. Phys. A: Math. Theor.},
   volume       = {51},
   pages        = {414009},
   doi          = {10.1088/1751-8121/aabe9c},
}

@ARTICLE{Chen16,
   author       = {{L. Chen} and {L. Yu}},
   title        = {Entanglement cost and entangling power of bipartite unitary and permutation operators},
   year         = {2016},
   journal      = {Phys. Rev. A},
   volume       = {93},
   pages        = {042331},
   doi          = {10.1103/PhysRevA.93.042331},
}

@ARTICLE{Soeda13,
   author  = {{A. Soeda} and {M. Murao}},
   title   = {Comparing the globalness of bipartite unitary operations: delocalisation power, entanglement cost and entangling power},
   year    = {2013},
   journal = {Mathematical Structures in Computer Science},
   volume  = {23},
   pages   = {454--470},
}

@BOOK{Nielsen10,
   author       = {{M. A. Nielsen} and {I. L. Chuang}},
   year         = {2010},
   title        = {Quantum Computation and Quantum Information: 10th Anniversary Edition},
   publisher    = {Cambridge University Press}
}

@BOOK{Watrous18,
   author       = {{J. Watrous}},
   year         = {2018},
   title        = {The Theory of Quantum Information},
   publisher    = {Cambridge University Press}
}

@ARTICLE{Aharonov03,
   author       = {{D. Aharonov}},
   title        = {A simple proof that Toffoli and Hadamard are quantum universal},
   year         = {2003},
   journal      = {arXiv:quant-ph/0301040},
}

@ARTICLE{Shor97,
   author       = {{P. W. Shor}},
   title        = {Polynomial-time algorithms for prime factorization and discrete logarithms on a quantum computer},
   year         = {1997},
   journal      = {SIAM J. Comput.},
   volume       = {26},
   pages        = {1484--1509},
}

@ARTICLE{Grover96,
   author       = {{L. K. Grover}},
   title        = {A fast quantum mechanical algorithm for database search},
   year         = {1996},
   journal      = {In Proc. Twenty-Eighth Annual ACM Symposium on Theory of Computing},
   pages        = {212--219},
}

@ARTICLE{Feynman82,
   author       = {{R. P. Feynman}},
   title        = {Simulating physics with computers},
   year         = {1982},
   journal      = {Int. J. Theor. Phys.},
   volume       = {21},
   pages        = {467--488},
}

@ARTICLE{Lloyd96,
   author       = {{S. Lloyd}},
   title        = {Universal quantum simulators},
   year         = {1996},
   journal      = {Science},
   volume       = {273},
   pages        = {1073--1078},
}

@ARTICLE{Kassal08,
   author       = {{I. Kassal} and {S. P. Jordan} and {P. J. Love} and {M. Mohseni} and {A. Aspuru-Guzik}},
   title        = {Polynomial-time quantum algorithm for the simulation of chemical dynamics},
   year         = {2008},
   journal      = {Proc. Natl. Acad. Sci. U.S.A.},
   volume       = {105},
   pages        = {18681--18686},
}

@ARTICLE{Renou21,
   author       = {{M.-O. Renou} and {D. Trillo} and {M. Weilenmann} and {T. P. Le} and {A. Tavakoli} and {N. Gisin} and {A. Ac{\'\i}n} and {M. Navascu{\'e}s}},
   title        = {Quantum theory based on real numbers can be experimentally falsified},
   year         = {2021},
   journal      = {Nature},
   volume       = {600},
   pages        = {625--629},
}

@ARTICLE{Wu24,
   author       = {{K.-D. Wu} and {T. V. Kondra} and {C. M. Scandolo} and {S. Rana} and {G.-Y. Xiang} and {C.-F. Li} and {G.-C. Guo} and {A. Streltsov}},
   title        = {Resource theory of imaginarity in distributed scenarios},
   year         = {2024},
   journal      = {Commun. Phys.},
   volume       = {7},
   pages        = {171},
}

@ARTICLE{Kitaev97,
   author       = {{A. Y. Kitaev}},
   title        = {Quantum computations: Algorithms and error correction},
   year         = {1997},
   journal      = {Russ. Math. Surv.},
   volume       = {52},
   pages        = {1191--1249},
}

@ARTICLE{Shi02,
   author       = {{Y. Shi}},
   title        = {Both Toffoli and Controlled-NOT need little help to do universal quantum computation},
   year         = {2002},
   journal      = {arXiv:quant-ph/0205115},
}

@ARTICLE{Amy23,
   author       = {{M. Amy} and {A. N. Glaudell} and {S. M. Li} and {N. J. Ross}},
   title        = {Improved synthesis of Toffoli-Hadamard circuits},
   year         = {2023},
   journal      = {Reversible Computation (LNCS)},
   volume       = {13960},
   pages        = {169--209},
}

@ARTICLE{gour2020dynamicalresources,
   author       = {{G. Gour} and {C. M. Scandolo}},
   title        = {Dynamical resources},
   year         = {2021},
   journal      = {arXiv:2101.01552},
}

@ARTICLE{Hokkyo26,
   author       = {{A. Hokkyo} and {H. Tajima}},
   title        = {Quantitative {Wigner-Araki-Yanase} {Theorems} for {Unitary} and {Antiunitary} {Symmetries}},
   year         = {2026},
   journal      = {arXiv:2607.09075},
}



\section*{Acknowledgements}


We thank Koji Azuma for fruitful discussions.
The authors received no specific funding for this work.

\section*{Author contributions statement}

S.A. conceived the key idea through discussions with Y.N. and Y.T.
Y.N. refined the key idea through discussions with Y.T. and S.A.
Y.N. performed the calculations and showed the theorem and proposition, and all authors contributed to checking their validity and writing the manuscript.

\section*{Additional information}



\noindent \textbf{Competing interests:} The authors declare no competing interests.

\par\medskip
{\small This version of the article has been accepted for publication, after peer review, but is not the Version of Record and does not reflect post-acceptance improvements or corrections. The Version of Record is available online at: \href{https://doi.org/10.1038/s41598-026-70782-1}{https://doi.org/10.1038/s41598-026-70782-1}.}




\end{document}